\documentclass[ reprint, amsmath,amssymb, aps,pra]{revtex4-1}
\usepackage{tikz}
\usepackage{amscd}
\usepackage{amsthm}
\usepackage{graphicx}
\usepackage{mathdots}
\usepackage{epstopdf}
\usepackage{enumerate}
\usepackage{changepage}
\usepackage{lipsum}

\usepackage{amsmath}
\usepackage{amssymb}
\usepackage{amsfonts}
\usepackage{enumerate}
\usepackage{mfirstuc}
\usepackage{mathrsfs}
\usepackage{graphicx}
\usepackage{epstopdf}
\usepackage{enumerate}
\usepackage{changepage}

\usepackage{tikz}
\usepackage{amscd}
\usepackage{amsthm}
\usepackage{graphicx}
\usepackage{mathdots}
\usepackage{epstopdf}
\usepackage{setspace}
\usepackage{fancyhdr}
\usepackage{lipsum}
\usepackage{enumerate}
\usepackage{changepage}
\usepackage[misc]{ifsym}
\usepackage{bbding}

\usepackage[
colorlinks,
linkcolor = blue,
citecolor = blue,
urlcolor = blue]{hyperref}
\def \qed {\hfill \vrule height7pt width 7pt depth 0pt}

\newtheorem{theorem}{Theorem}
\newtheorem{corollary}{Corollary}
\newtheorem{lemma}{Lemma}
\newtheorem{example}{Example}
\newtheorem{definition}{Definition}

\usepackage{dcolumn}
\usepackage{bm}

\begin{document}


\title{ Special unextendible  entangled bases with continuous  integer cardinality}

\date{\today}

\author{Yan-Ling Wang}%
\email{wangylmath@yahoo.com}
\affiliation{
	School of Computer Science and Network Security, Dongguan University of Technology, Dongguan 523808, China
}

\begin{abstract}
  Special unextendible  entangled basis of ``type $k$"  (SUEBk),  a set of  incomplete
 orthonormal special  entangled states of ``type $k$"   whose complementary space has no special  entangled state  of ``type $k$".  This concept can be seem as a generalization of the unextendible product basis (UPB) introduced by Bennett et al. in  [\href{https://doi.org/10.1103/PhysRevLett.82.5385}{Phys. Rev. Lett. \textbf{82}, 5385(1999)}] and the unextendible maximally
 entangled basis (UMEB)  introduced by  Bravyi and Smolin in [\href{https://journals.aps.org/pra/abstract/10.1103/PhysRevA.84.042306}{Phys. Rev. A  \textbf{84}, 042306(2011)}].   We present an efficient method to construct  sets   of  SUEBk.   The main strategy here is to decompose the whole space into two subspaces such that the rank of one subspace  can be easily upper bounded  by $k$ while  the other one can be generated by two kinds of the special entangled states of type $k$.   This method is very effective for those $k=p^m\geq 3$ where $p$ is a prime number. For these cases, we can otain sets of    SUEBk  with continuous  integer cardinality when the local dimensions are large. Moreover, one can find that our method here can be easily  extended  when there are more than two kinds of the special entangled states of type $k$ at hand. 
	\begin{description}
\item[PACS numbers] 03.67.Hk,03.65.Ud
\end{description}
 \end{abstract}                            
\maketitle

\section{Introduction}


Unextendible product basis (UPB)\cite{Bennett99,DiVincenzo03},  a set of incomplete
orthonormal product states whose complementary space has no product states, has been shown to be useful for constructing bound entangled states and displaying quantum nonlocality without entanglement \cite{Ben99,Ran04,Hor03}.

As anology of the UPB,  Bravyi and Smolin introduced the unextendible maximally entangled basis (UMEB) \cite{Bravyi11}, a set of orthonormal maximally entangled states in
$\mathbb{C}^{d}\bigotimes\mathbb{C}^{d}$ consisting of fewer than $d^2$ vectors which have no additional maximally
entangled vectors orthogonal to all of them.  There they proved that no UMEB exists in two qubits system and presented examples of   UMEBs in $\mathbb{C}^{3}\bigotimes\mathbb{C}^{3}$ and $\mathbb{C}^{4}\bigotimes\mathbb{C}^{4}$. Since then, the UMEB was   further studied by  several reseachers \cite{Chen13, Li14,Wang14,Nan15,Wang17,Guo15,Guo16,Zhang18}. Lots  of the works pay attention to  the UMEBs for general quantum systems  $\mathbb{C}^{d}\bigotimes\mathbb{C}^{d'}$.
The cardinality of the constructed UMEBs are always multiple of $d$ or $d'$.
 
Guo \emph{et al.} extended these two concepts to the states with fixed Schmidt numbers and studied the complete basis \cite{Guo152} and the unextendible ones \cite{Guo14}. There they introduced the notion of  special entangled states of type $k$ (SESk): an entangled state whose nonzero Schmidt coefficients are all equal to $1/\sqrt{k}$. Then a special unextendible entangled basis of type $k$ (SUEBk)is  a set of orthonormal SESk in
$\mathbb{C}^{d}\bigotimes\mathbb{C}^{d'}$ consisting of fewer than $dd'$ vectors which have no additional SESk orthogonal to all of them.  Quite rencently, there are several results related to this subject \cite{Shi19,Yong19}.  Similar to the UMEBs, the cardinality of most of the  known SUEBk's are multiple of $k$. 

 Therefore, it is interesting to ask whether there are SUEBks with continuous  integer cardinality or not. Inspired by the technique used in \cite{Li19}, we try to study  this question in this paper.

The remaining of this article is organized as follows.   In Sec. \ref{second}, we first introduce the  concept of special unextendible entangled basis and its equivalent form in matrix settings.  In Sec. \ref{third},  we present our main idea to construct the SUEBk. In Sec. \ref{fourth} and \ref{fifth},    based on the combinatoric concept: weighing matrices, we give two kinds of constructions of SUEBk whose cardinality varying in a  consecutive integer set.     Finally, we draw the conclusions and put forward some interesting questions in the last section.

\section{ Preliminaries}\label{second}
Let $[n]$ denote the set $\{1,2,\cdots, n\}$.   Let $\mathcal{H}_A$, $\mathcal{H}_B$ be  Hilbert spaces of dimension $d$ and $d'$ respectively. It is well known that any bipartite pure state in $\mathbb{C}^{d}\otimes \mathbb{C}^{d'}$ has a Schmidt decomposition. That is, any unit vector  $|\phi\rangle$ in $\mathbb{C}^{d}\otimes \mathbb{C}^{d'}$ can be written as
$$|\phi\rangle=\sum_{i=1}^k \lambda_i |e_i\rangle_A|e_i\rangle_B,     \  \ \sum_{i=1}^k \lambda_i^2=1$$
where $\lambda_i>0$ and $\{|e_i\rangle_A\}_{i=1}^k$ ($\{|e_i\rangle_B\}_{i=1}^k$) are orthonormal states of system $A$ (resp. $B$). The number $k$ is known as the Schmidt number of $|\phi\rangle $ and  we  denote it by $S_r(\phi)$. The set $\Lambda(|\phi\rangle):=\{\lambda_i\}_{i=1}^k$  is called the nonzero Schmidt coefficients of $|\phi\rangle.$ If all these $\lambda_i$s are equal to $1/\sqrt{k}$, we call $|\phi\rangle$ a special entangled state  of type $k$ ($2\leq k\leq d$). And we denote the set of all the special entangled states  of type $k$ to be SESk.  One notice that SESk is exactly the set of maximally entangled states  in $\mathbb{C}^d\otimes \mathbb{C}^{d'}$ when $k=\min\{d,d'\}.$

\begin{definition}[See \cite{Guo152}]\label{SUEntangledbasis}
 A set of states $\{|\phi_i\rangle\}_{i=1}^{n}$  ($1\leq n\leq dd'-1$) in $\mathbb{C}^d\otimes \mathbb{C}^{d'}$ is called a special unextendible entangled basis of type  $k$ \emph{(SUEBk)} if
 \begin{enumerate}[(1)]
 \item  $\langle \phi_i|\phi_j\rangle=\delta_{ij}, i,j\in [n];$
 \item   $|\phi_i\rangle\in$  SESk for all $i\in [n]$;
   \item  If $\langle\phi_i|\phi\rangle=0$ for all $i\in [n]$, then $|\phi\rangle \notin$  SESk.
 \end{enumerate}
  \end{definition}
The concept SUEBk generalizes the UPB  ($k=1$) and the UMEB  ($k=d$). In order to study SUEBk, it is useful to consider its matrix form.
Let $|\phi\rangle$ be  a  pure quantum states  in $ \mathcal{H}_A\otimes \mathcal{H}_B$. Under the computational bases $\{|i\rangle_A\}_{i=1}^{d}$ and $\{|j\rangle_B\}_{j=1}^{d'}$, it can be expressed as
$$|\phi\rangle=\sum_{i=1}^d\sum_{j=1}^{d'} m_{ij}^\phi |i\rangle_A|j\rangle_B.$$
We call the $d\times d'$ matrix $M_{\phi}:=(m_{ij}^\phi)$ the corresponding matrix representation of $|\phi\rangle.$ The correspondence is good in the following sense:
 \begin{enumerate}[(1)]
 \item  Inner product preserving: $$\langle \psi|\phi\rangle= \sum_{i=1}^{d}\sum_{j=1}^{d'}\overline{m_{ij}^\psi} m_{ij}^\phi=\text{Tr}(M_{\psi}^\dagger M_{\phi})=\langle M_{\psi}, M_{\phi}\rangle;$$
  \item  Schmidt number corresponding  to the matrix rank: $S_r(|\phi\rangle)=\text{rank}(M_{\phi});$
   \item   Nonzero Schmidt coefficients corresponding  to the nonzero singular values.
 \end{enumerate}

 With this correspondence, we can restate the concept  in definition   \ref{SUEntangledbasis} as follows.

\begin{definition}\label{SUEntangledbasis}
 A set of matrices $\{M_i\}_{i=1}^{n}$  ($1\leq n\leq dd'-1$) in $\text{Mat}_{d\times d'}(\mathbb{C})$ is called a special unextendible singular values basis with    nonzero  singular values being $\{1/\sqrt{k}\}$ \emph{(SUSVBk)} if
 \begin{enumerate}[(1)]
 \item  $\langle M_i,M_j\rangle=\delta_{ij}, i,j\in [n];$
 \item   The nonzero singular values of  $M_i$ are all equal to  $1/\sqrt{k}$  for each $i\in [n]$;
   \item  If $\langle M_i, M \rangle=0$ for all $i\in [n]$,  then some nonzero singular value of $M$ do not equal to $1/\sqrt{k}$.
 \end{enumerate}
  \end{definition}

  Due to the good   correspondence of the states and matrices,   $\{|\psi_i\rangle\}_{i=1}^n$ is a set of SUEBk in $\mathbb{C}^d\otimes \mathbb{C}^{d'}$ if and only if $\{M_{\psi_i}\}_{i=1}^n$  is a set of SUSVBk in $\text{Mat}_{d\times d'}(\mathbb{C})$. Therefore, in order to construct a set of $n$ members SUEBk in $\mathbb{C}^d\otimes \mathbb{C}^{d'}$, it is sufficient to   construct a set of $n$ members SUSVBk in $\text{Mat}_{d\times d'}(\mathbb{C})$.

\vskip 5pt

\section{strategy for constructing susvbk}\label{third}

{\bf \emph{Observation 1}}-- It is uneasy to calculate the singular values of an arbitrary matrix. However,  if there are only $k$ nonzero elements  in $M$ (say $m_{i_1,j_1},\cdots,m_{i_k,j_k} $) and these elements happen to be in different rows and columns, then there are exactly $k$ nonzero singular values of $M$ and they are just $|m_{i_1,j_1}|, \cdots, |m_{i_k,j_k}|$.
For example, let $M$ be
$$
\left[\begin{array}{cccccc}
\frac{1}{\sqrt{2}} & 0 & 0& 0& 0&0\\
0 &0 &0&\frac{\sqrt{-1}}{\sqrt{3}}&0&0\\
0 &0 &\frac{1}{\sqrt{12}}&0&0&0\\
0 &\frac{w}{\sqrt{24}}&0&0&0&0\\
0 &0 &0&0&\frac{1}{\sqrt{24}}&0\\
0 &0 &0&0&0&0\\
0 &0 &0&0&0&0
\end{array}
\right]
$$
where $w=e^{2\pi \sqrt{-1}/3}$. Then the nonzero singular values of $M$ are $\frac{1}{\sqrt{2}}$,  $\frac{1}{\sqrt{3}}$,
$\frac{1}{\sqrt{12}}$, $\frac{1}{\sqrt{24}}$, $\frac{1}{\sqrt{24}}$.

\vskip 5pt

{\bf \emph{Observation 2}}-- If there are exactly $k$ nonzero singular values of $M$, then the  rank of $M$ is $k$.  Therefore, if one can prove that $\text{rank}(M)<k$, then $M$ cannot be a matrix with $k$ nonzero singular values.

\vskip 5pt

With the two observations above, our strategy for constructing  an $n$-members SUSVBk can be roughly described by two steps. Firstly, we construct  a set of $n$-members of orthnormal matrices $\mathcal{M}:=\{M_i\}_{i=1}^n$ such that there are exactly $k$ nonzero elements  in $M_i$  whose modulos are all $1/\sqrt{k}$ and these elements happen to be in different rows and columns. Secondly, we need to show that the rank of  any matrix in the  complementary space
of $\mathcal{M}$ (define as $\mathcal{M}^\perp:=\{M\in\text{Mat}_{d\times d'}(\mathbb{C}) | \langle M_i, M\rangle= 0, \forall  M_i\in \mathcal{M}\} $) is less than $k$.

	


Let $d,d'$ be integers  such that $ 2 \leq d \leq d'$. We define the coordinate set to be
$$ \mathcal{C}_{d\times d'}:=\{(i,j)\in \mathbb{N}^2| i\in[d], j\in[d']\}.$$
Now we define an order   for the set $ \mathcal{C}_{d\times d'}$. Equivalently, we can define a bijection:
$$\begin{array}{crcl}
\mathcal{O}_{d\times d'}:& \mathcal{C}_{d\times d'}&\longrightarrow& [dd']\\[2mm]
            &  (i,j)&\longmapsto&\left\{
             \begin{array}{ll}
          (j-i)d+i &\text{ if } i\leq j;\\
           (d'+j-i)d+i &\text{ if } i>j.
             \end{array}\right\}
                       \end{array}
            $$
            Then we call ($\mathcal{C}_{d\times d'},\mathcal{O}_{d\times d'})$  an ordered set (See Fig. \ref{orderfig} for an example). We can also define an order $\mathcal{O}_{d\times d'}$ for the cases    $d'\leq d$ by   $\mathcal{O}_{d\times d'}:= \mathcal{O}_{d'\times d}$.

\begin{figure}[h]
		\includegraphics[width=0.45\textwidth,height=0.2\textwidth]{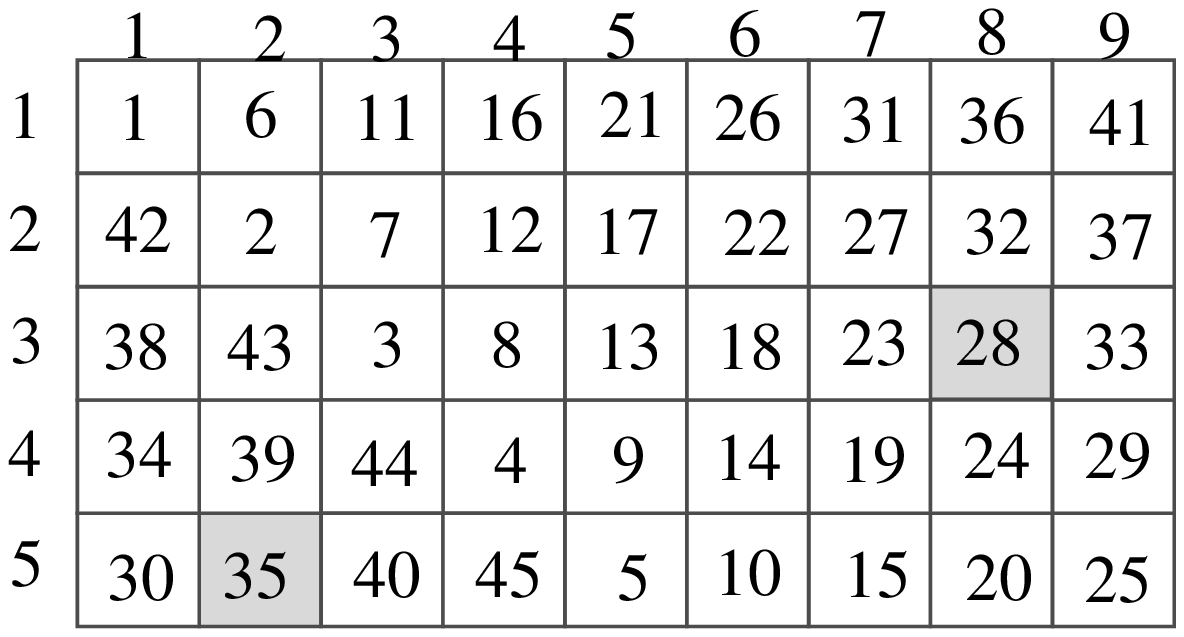}
		\caption{This is a picture of the order $\mathcal{O}_{5\times 9}$ on the coordinate set $\mathcal{C}_{5\times 9}$. For examples,  $\mathcal{O}_{5\times 9}[(3,8)]=(8-3)\times 5+3=28$, and  $\mathcal{O}_{5\times 9}[(5,2)]=(9+2-5)\times 5+5=35$. }\label{orderfig}
		\end{figure}

Let $(i_1,j_1), (i_2,j_2)$ be two different coordinates in $\mathcal{C}_{d\times d'} $. It is easy to check that if  $i_1=i_2$ or $j_1=j_2$, then $\big|\mathcal{O}_{d\times d'}[(i_1,j_1)]-\mathcal{O}_{d\times d'}[(i_2,j_2)]\big|\geq d-1$.   Therefore, any $d-1$  consecutive coordinates in $\mathcal{C}_{d\times d'}$ under the order $\mathcal{O}_{d\times d'}$ is  coordinately different. That is, these  $d-1$ coordinates must come from different rows and  different columns.

Let $P\subseteq \mathcal{C }_{d\times d'}$. Then  $P$ inherit an order $\mathcal{O}$ from that of $\mathcal{C }_{d\times d'}$(An order here means a bijective map from $P$ to $[\#P]$ where   $\# P$ to denote the number of elements in the  set $P$).   In fact, as $\# \mathcal{O}_{d\times d'}(P)=\# P$, there is an unique map $\pi_P$ from the set $\mathcal{O}_{d\times d'}(P)$ to $[\#P]$ which preserve the order of the numbers. Then we define $\mathcal{O}:= \pi_P   \mathcal{O}_{d\times d'}|_{P} $ to be the order of $P$  inherit   from that of $\mathcal{C }_{d\times d'}$. For example, let $P:=\{(1,2),(4,3), (5,6)\}\subseteq  \mathcal{C }_{5\times 9}$.  Then the  $\pi_P$ from the set $\{\mathcal{O}_{5\times 9}[(1,2)]=6, \mathcal{O}_{5\times 9}[(4,3)]=44, \mathcal{O}_{5\times 9}[(5,6)]=10\}$ to $[3]=\{1,2,3\}$ is just defined by: $\pi_P(6)=1, \pi_P(10)=2, \pi_P(44)=3.$ Therefore,  the order $\mathcal{O}$ of $P$  inherited  from  that of $\mathcal{C }_{d\times d'}$ is exactly the map:  $\mathcal{O}[(1,2)]=1,\mathcal{O}[(5,6)]=2,\mathcal{O}[(3,4)]=3.$

\vskip 5pt
In order to step forward, we first state the following observation  which is helpful for determine the orthogonality of matrices.
Let $P\subseteq \mathcal{C }_{d\times d'}$ and denote $\mathcal{O}$ the order of $P$ inherit from the $\mathcal{O }_{d\times d'}$.
$l$ denotes the number of elements in $P$.  As we have defined an order for the set $\mathcal{C}_{d\times d'}$, it reduces an order to its subset $P$. For any   vector $v \in \mathbb{C}^{l}$, we define
$$M_{d\times d'}(P,v):= \sum_{(i,j)\in P} v_{\mathcal{O}[(i,j)]}E_{i,j}$$  
where $E_{i,j}$ denote the $d\times d'$ matrix whose $(i,j)$ coordinate is 1 and zero  elsewhere.
\vskip 5pt

\begin{lemma}\label{Orthogonal}  Let $P_1,P_2\subseteq \mathcal{C}_{d\times d'}$ be nonempty sets and $v,w$ be vectors of dimensions $\#P_1$ and $\#P_2$ respectively.  Then we have the following statements.
			\begin{enumerate}[(a)]
				\item   If $P_1\cap P_2=\emptyset$, then we have $$\langle M_{d\times d'}(P_1,v), M_{d\times d'}(P_2,w)\rangle=0.$$
				\item   If $P_1= P_2$ and $v,w$ are orthogonal to each other, then we also  have $$\langle  M_{d\times d'}(P_1,v),  M_{d\times d'}(P_2,w)\rangle=0.$$
			\end{enumerate}
\end{lemma}
\begin{proof}
	Denote $\mathcal{O}_1$ and $\mathcal{O}_2$ the orders of $P_1$ and $P_2$ inherit from the $\mathcal{O }_{d\times d'}$ respectively. 	
		\begin{enumerate}[(a)]
			
				\item As 
				$$
				\begin{array}{ccl}
				M_{d\times d'}(P_1,v):&=& \sum_{(i,j)\in P_1} v_{\mathcal{O}_1[(i,j)]}E_{i,j}, \\ M_{d\times d'}(P_2,v):&=& \sum_{(k,l)\in P_2} w_{\mathcal{O}_2[(k,l)]}E_{k,l},
				\end{array}$$
				we have
$$\begin{array}{rl}
&\langle M_{d\times d'}(P_1,v), M_{d\times d'}(P_2,w)\rangle\\[1mm]
=&\text{Tr}[M_{d\times d'}(P_1,v)^\dagger M_{d\times d'}(P_2,w) ]\\[1mm]
=&\displaystyle\sum_{(i,j)\in P_1}\displaystyle\sum_{(k,l)\in P_2}  \overline{v_{\mathcal{O}_1[(i,j)]}}w_{\mathcal{O}_2[(k,l)]} \text{Tr}[E_{j,i}E_{k,l}]\\[1mm]
=&\displaystyle\sum_{(i,j)\in P_1}\displaystyle\sum_{(k,l)\in P_2}  \overline{v_{\mathcal{O}_1[(i,j)]}}w_{\mathcal{O}_2[(k,l)]} \delta_{ik}\delta_{jl}=0.
\end{array}
$$
The last equality holds as  the condition $P_1\cap P_2=\emptyset$ implies $\delta_{ik}\delta_{jl}=0$.
	\item For the second part, we have the following equalities: $$\begin{array}{rl}
&\langle M_{d\times d'}(P_1,v), M_{d\times d'}(P_1,w)\rangle\\[1mm]
=&\text{Tr}[M_{d\times d'}(P_1,v)^\dagger M_{d\times d'}(P_1,w) ]\\[1mm]
=&\displaystyle\sum_{(i,j)\in P_1}\displaystyle\sum_{(k,l)\in P_1}  \overline{v_{\mathcal{O}_1[(i,j)]}}w_{\mathcal{O}_1[(k,l)]} \text{Tr}[E_{j,i}E_{k,l}]\\[1mm]
=&\displaystyle\sum_{(i,j)\in P_1}\displaystyle\sum_{(k,l)\in P_1}  \overline{v_{\mathcal{O}_1[(i,j)]}}w_{\mathcal{O}_1[(k,l)]} \delta_{ik}\delta_{jl}\\[1mm]
=&\displaystyle\sum_{(i,j)\in P_1}   \overline{v_{\mathcal{O}_1[(i,j)]}}w_{\mathcal{O}_1[(i,j)]}=\langle v|w\rangle=0. \\[1mm]
\end{array}
$$	\end{enumerate}
\end{proof}

\section{{First Type  of suebk }}\label{fourth}

In the following, we try to construct a set of matrices $\mathcal{M}:=\{M_i\}_{i=1}^n$ which generates all the matrices of the form $T_1$. While its complementary space   $\mathcal{M}^\perp$ is the set of matrices of the form $T_2$.
{\small 
	$$
T_1=\left[\begin{array}{ccccccc}
*&*&\cdots  & *&*&\cdots &*\\
\vdots&\vdots&\vdots & \vdots& \vdots&\vdots&\vdots \\
*&*&\cdots  & *&*&\cdots &*\\
0&0&\cdots  & 0&*&\cdots &*\\
0&0&\cdots  & 0&0&\cdots&0 \\
\vdots&\vdots&\vdots & \vdots& \vdots&\vdots&\vdots \\
0&0&\cdots  & 0&0&\cdots&0
\end{array}
\right],\  T_2=\left[\begin{array}{ccccccc}
0&0&\cdots  & 0&0&\cdots &0\\
\vdots&\vdots&\vdots & \vdots& \vdots&\vdots&\vdots \\
0&0&\cdots  & 0&0&\cdots &0\\
*&*&\cdots  & *&0&\cdots &0\\
*&*&\cdots  & *&*&\cdots&* \\
\vdots&\vdots&\vdots & \vdots& \vdots&\vdots&\vdots \\
*&*&\cdots  & *&*&\cdots&* \\
\end{array}
\right].
$$
}

\begin{example}\label{example1}
	 There exists  a SUEB3 in $\mathbb{C}^7 \otimes \mathbb{C}^7 $ whose cardinality is $47$.
\end{example}
\noindent\emph{Proof.} As $47=7\times 7-2$,  we  define $\mathcal{B}_{47}$ to be the  set with 47 elements which can be obtained by    deleting    $\{ (7,1),(7,2)\}$   from $\mathcal{C}_{7\times 7}$.  We can define an order $\mathcal{O}$ for the set $\mathcal{B}_{47}$.  In fact, the $\mathcal{O}$ is chosen to be the order of $\mathcal{B}_{47}$ inherited from that of $\mathcal{C}_{7\times 7}$ (See the left figure of Fig. \ref{7times7Figure} for an intuitive view).  Any $5$   consecutive elements of $\mathcal{B}_{47}$ under the order   $\mathcal{O}$    come from different rows and columns.  Firstly, we have the following identity
\begin{equation}\label{Large_7times7_decom}
47=9\times 3+ 5\times 4.
\end{equation}
Since there are $47$ elements in the set $\mathcal{B}_{47}$, by the decomposition (\ref{Large_7times7_decom}), we can divide the set $\mathcal{B}_{47}$ into $(9+5)$ sets: $9$ sets (denote by $S_i,  1\leq i\leq 9$)  of cardinality $3$ and $5$ sets (denote by $L_j, 1\leq j\leq 5$)  of cardinality $4$. In fact, we can divide $\mathcal{B}_{47}$ into these $14$ sets through its order $\mathcal{O}$. That is,
$$\begin{array}{l}
S_{i}:=\{\mathcal{O}^{-1}[3(i-1)+x]\ \big| \   x=1,\cdots, 3\}, 1\leq i\leq 9,\\[2mm] 
L_{j}:=\{\mathcal{O}^{-1}[27+4(j-1)+y]\ \big| \   y=1,\cdots,4\},  1\leq j\leq 5.
\end{array}$$
See the right hand side of Fig. \ref{7times7Figure} for an intuitive view of the set $S_i,L_j$.
Set  
$$
\text{CH}_3=
\left[\begin{array}{lll}
1 & 1 & 1 \\
1 & w & w^2  \\
1 & w^2 & w 
\end{array}
\right],\ \ 
O_4=\left[\begin{array}{rrrr}
0 & 1 & 1 & 1 \\
1 & 0 & -1 & 1 \\
1 & 1 & 0 & -1 \\
1 & -1 & 1 & 0
\end{array}
\right] 
$$
where $w=e^{\frac{2\pi \sqrt{-1}}{3}}$. 
We can easily check that
$
\text{CH}_3 \text{CH}_3^\dagger=3I_3$ and $
O_4 O_4^\dagger= 3I_4.
$
Now set $v_x$  to be the $x$-th row of $\text{CH}_3$  ($x=1,2,3$) and $w_y$ to be the $y$-th row of $O_4$  ($y=1,2,3,4$). So $v_x\in \mathbb{C}^3$ and $v_y\in \mathbb{C}^4$. So we can construct the following $9\times 3+ 5\times 4=47$ matrices:
$$\begin{array}{c}
M_{7\times 7} (S_i, \frac{1}{\sqrt{3}}v_x), M_{7\times 7} (L_j,\frac{1}{\sqrt{3}}w_y),\\[2mm]
 \ 1\leq i\leq 9,1\leq x\leq 3, 1\leq j\leq 5,1\leq y\leq 4.
 \end{array}$$

Let $\mathcal{M}$ to be the set of the above matrices.  Note that  the elements of each $S_i$  or $L_j$  are coordinately different.  Hence by  {\bf Observation} 1, the states correspond  to the above $N$ matrices  belong to  SES3.
Since $\text{CH}_3 \text{CH}_3^\dagger=3I_3$,  $v_1,v_2,v_3$ are pairwise orthogonal. Similarly, as $O_4 O_4^\dagger =3 I_4$, $w_1,w_2,w_3,w_4$ are also pairwise orthogonal.  And the $s+t$ sets above are pairwise disjoint.
Therefore, by Lemma \ref{Orthogonal}, the  $47$ matrices above are pairwise orthogonal.  Set $V$ be the linear  space spaned by the matrices in $\mathcal{M}$. Each matrix in $\mathcal{B}_\perp:=\{E_{i,j}\in \text{Mat}_{7\times7}(\mathbb{C})| (i,j)\in \mathcal{C}_{7\times 7} \setminus \mathcal{B}_{47}\}$ is orthogonal to $V$.  And  the dimension of $\text{span}_\mathbb{C}(\mathcal{B}_\perp)$ is just $2$. Therefore, $V^\perp=\text{span}_\mathbb{C}(\mathcal{B}_\perp)$.  One should note that the rank of any nonzeto matrix in $\text{span}_\mathbb{C}(\mathcal{B}_\perp)$ is $1$.   Such a state cannot lie in SEB3.  Therefore, the set of states corresponding to the matrices $\mathcal{M}$  consists a SUEB3. \qed

\begin{figure}[h]
	\includegraphics[width=0.22\textwidth,height=0.19\textwidth]{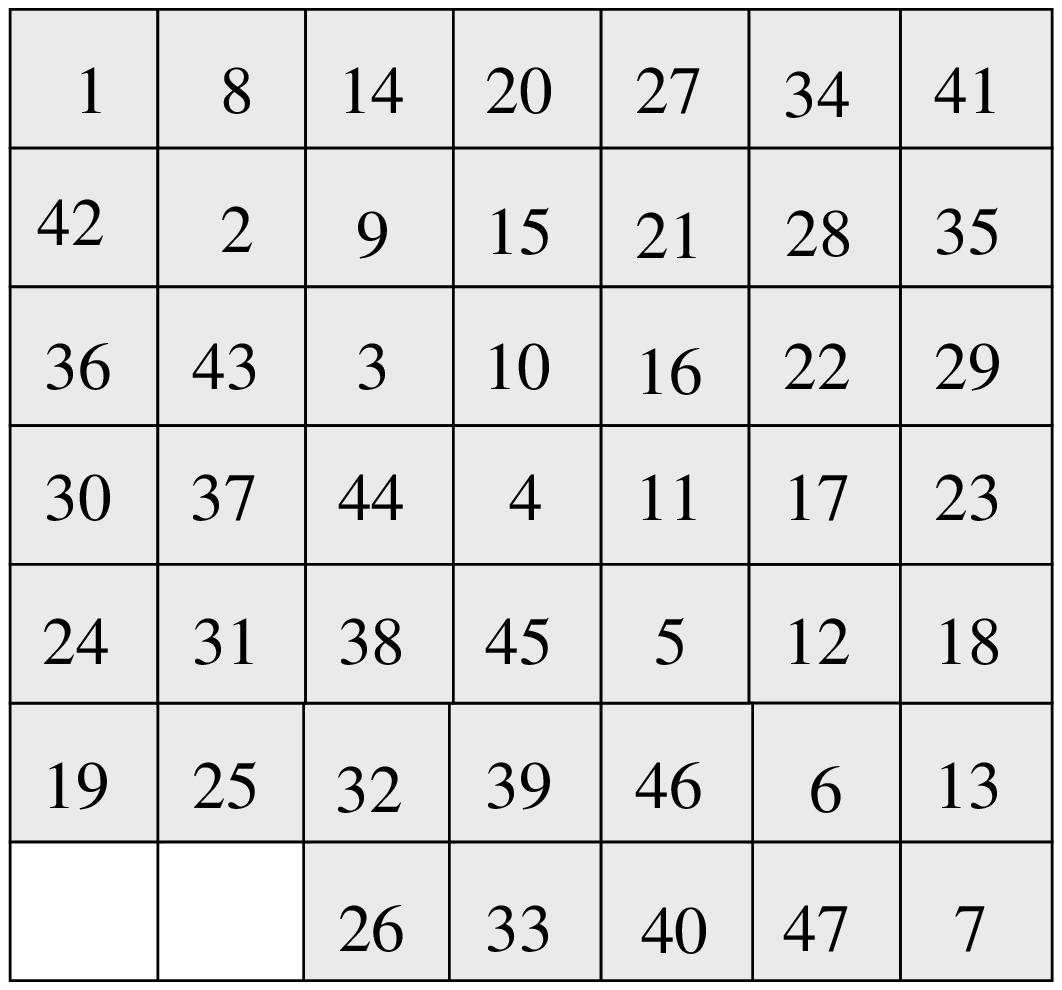}
	\includegraphics[width=0.24\textwidth,height=0.19\textwidth]{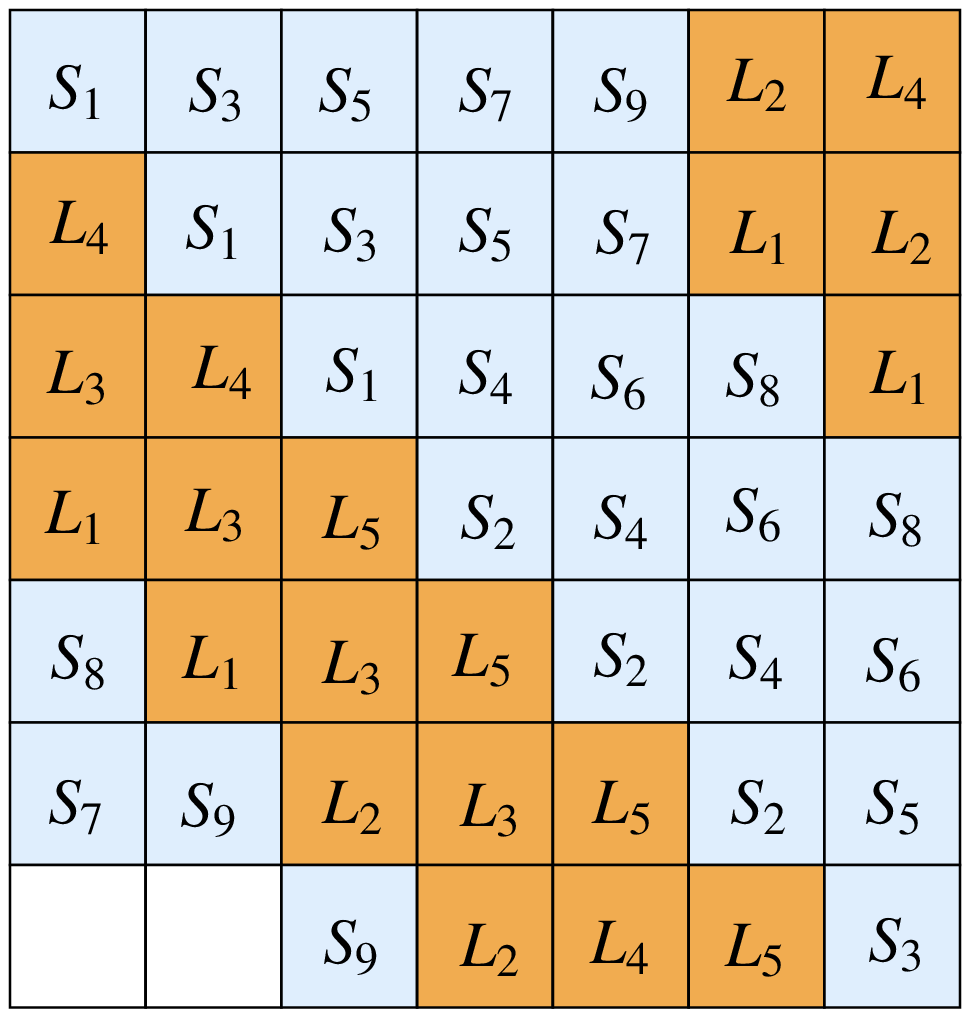}
	\caption{ The left figure shows the order of subset of $\mathcal{C}_{7\times 7}$. While the right hand one shows the distribution of the short and long states through this order. }\label{7times7Figure}
	
\end{figure}

\vskip 8pt

One can find that the $\text{CH}_3$ and $O_4$ play an important role in the proof of the   example \ref{example1}.   We give their generalizations by the following matrix and the weighing matrix in definition \ref{weigh_def}. 
There always exists some complex Hadamard matrix of order $d$. For example,
\begin{equation}\label{complex}
{\text{CH}}_d:=\left[
\begin{array}{lllll}
1 & 1 & 1 & \cdots & 1 \\[0.5mm]
1 & \omega_d  & \omega_d^2 & \cdots & \omega_d^{d-1} \\[0.5mm]
1 & \omega_d^2 & \omega_d^4 & \cdots  & \omega_d^{2(d-1)} \\[0.5mm]
\vdots & \vdots & \vdots & \ddots & \vdots \\
1 & \omega_d^{d-1} & \omega_d^{2(d-1)} & \cdots & \omega_d^{(d-1)^2}
\end{array}
\right],  
\end{equation}
 where $omega_d=e^{\frac{2\pi\sqrt{-1}}{d}}$. In fact, this is the Fourier $d$-dimensional matrix (discrete Fourier transform). The matrix $\text{CH}_d$ satisfies
\begin{equation}\label{weigh2}
\text{CH}_d \text{CH}_d^\dagger= dI_d.
\end{equation}

\begin{definition}[See \cite{Berman78}]\label{weigh_def}
	A generalized weighing matrix  is a square $a\times a$ matrix $A$  all of whose non-zero entries are $n$-th roots of unity such that $AA^\dagger=kI_a$. It follows that $\frac{1}{\sqrt{k}}A$ is a unitary matrix so that $A^\dagger A=kI_a$ and every row and column of $A$ has exactly $k$ nonzero entries. $k$ is called the weight and $n$ is called the order of $A$.   We denote $W(n,k,a)$ the set of all weight $k$ and order $a$ generalized weighing matrix whose nonzero entries being $n$-th root.
\end{definition}


One can find the following lemma via the theorem 2.1.1 on the book ``The Diophantine Frobenius Problem" \cite{Jorge06}. The related problem is also known as Frobenius coin problem  or coin problem.

\begin{lemma}[\cite{Jorge06}]\label{Frobenius}
	Let $a,b$ be positive integers and coprime. Then for every integer $N\geq (a-1)(b-1)$, there are non-negative integers $x,y$ such that $N=xa+yb$.
	
\end{lemma}

\begin{theorem}\label{USEBk_weigh}
	Let $k$ be a positive integer. 	 Suppose there exist $a,b,m,n \in \mathbb{N}$ such that $W(m,k,a)$ and $W(n,k,b)$ are nonempty and $\gcd(a,b)=1$.  If $d,d'$ are integers such that $d\geq  \max\{a,b\}+k$ and  $d'\geq\max\{a,b\}+1,$ then for any integer $N\in [(d-k+1)d', dd'-1]$, there exists  a SUEBk in $\mathbb{C}^d\otimes \mathbb{C}^{d'} $ whose cardinality is exactly $N$.
\end{theorem}

\noindent\emph{Proof.} Without loss of generality, we suppose $a< b$ and  $A\in W(m,k,a)$, $B\in W(n,k,b)$. Let $A_1,\cdots, A_a$ be the rows of $A$ and $B_1,\cdots, B_b$ be the rows of $B$. Any integer $N\in [(d-k+1)d', dd'-1]$ can be  written  uniquely as $N=d' q+ r$ where $ (d-k+1)\leq q\leq d-1$ and $r$ is an integer with $0\leq r<d'$ . Then we have a coordinate set $\mathcal{C}_{(q+1)\times d'}$  with order $\mathcal{O}_{(q+1)\times d'}$. Notice that any $q$  consecutive elements  of $\mathcal{C}_{(q+1)\times d'}$ under the order  $\mathcal{O}_{(q+1)\times d'}$ are  coordinate different. Denote $\mathcal{B}_N$ to be the set by deleting the elements $\{ (q+1,i) | 1\leq i \leq d'-r\}$  from  $\mathcal{C}_{(q+1)\times d'}$.  The subset $\mathcal{B}_N$ inherit an order $\mathcal{O}$ from that of $\mathcal{C}_{(q+1)\times d'}$.   As $\big| \mathcal{O}_{(q+1)\times d'}[(q+1,i)]-\mathcal{O}_{(q+1)\times d'}[(q+1,j)]\big| \geq q$ for any $1\leq i\neq j\leq d'$, any $q-1$  consecutive elements  of $\mathcal{B}_N$ under the order  $\mathcal{O}$ are  coordinate different.   Since $ q-1\geq d-k \geq \max\{a,b\}$, any $a$ or $b$  consecutive elements  of $\mathcal{B}_N$  under the order   $\mathcal{O}$  come from different rows and columns. As $N\geq qd'>(a-1)\times (b-1)$, by Lemma \ref{Frobenius}, there exist nonnegative integers  $s,t$  such that
\begin{equation}\label{Large_dtimesd_decom}
N=s\times a+ t\times b.
\end{equation}
Since there are $N$ elements in the set $\mathcal{B}_{N}$, by the decomposition (\ref{Large_dtimesd_decom}), we can divide the set $\mathcal{B}_{N}$ into $(s+t)$ sets: $s$ sets (denote by $S_i,  1\leq i\leq s$)  of cardinality $a$ and $t$ sets (denote by $L_j, 1\leq j\leq t$)  of cardinality $b$. In fact, we can divide $\mathcal{B}_{N}$ into these $s+t$ sets through its order $\mathcal{O}$. That is,
$$\begin{array}{rl}
	S_{i}:=&\{\mathcal{O}^{-1}[(i-1){a}+x]\ \big| \   x=1,\cdots, a\}, 1\leq i\leq s,\\[2mm] 
	L_{j}:=&\{\mathcal{O}^{-1}[{s}a+(j-1)b+y]\ \big| \   y=1,\cdots, b\},  1\leq j\leq t.
	\end{array}$$

Then we can construct the following $s\times a+ t\times b=N$ matrices:
$$\begin{array}{c}
\{M_{d\times d'} (S_i, \frac{1}{\sqrt{k}}A_x), M_{d\times d'} (L_j,\frac{1}{\sqrt{k}}B_y) \\[2mm]
\ 1\leq i\leq s,1\leq x\leq a, 1\leq j\leq t,1\leq y\leq b.
\end{array}$$

Let $\mathcal{M}$ to be the set of the above matrices. Note that the $(s+t)$ sets $S_1,\cdots,S_s,L_1\cdots,L_t$ are pairwise disjoint. And the rows of $A$ (resp. $B$) are orthogonal to each other  as  $AA^\dagger= kI_a$ (resp. $BB^\dagger=kI_b$). By  Lemma \ref{Orthogonal}, the above $sa+tb$ matrices are orthogonal to each other. By construction, all the sets $S_1\dots,S_s,L_1,\cdots, L_t$ are  all coordinately different.  Using this fact and the  definition of generalized weighing matrices, the states corresponding to these matrices are all belong to SESk (see {\bf Observation} 1).  Set $V$ be the linear subspace of $\text{Mat}_{d\times d'}(\mathbb{C})$. Each matrix in $\mathcal{B}_\perp:=\{E_{i,j}\in\text{Mat}_{d\times d'}(\mathbb{C})| (i,j)\in \mathcal{C}_{d\times d'} \setminus \mathcal{B}_N\}$ is orthogonal to $V$.  And  the dimension of $\text{span}_\mathbb{C}(\mathcal{B}_\perp)$ is just $dd'-N$. Therefore, $V^\perp=\text{span}_\mathbb{C}(\mathcal{B}_\perp)$.  One should note that the rank of any matrix in $\text{span}_\mathbb{C}(\mathcal{B}_\perp)$ is less than $k$. That is to say, any state orthogonal to the states corresponding to $\mathcal{M}$ has Schmidt rank at most $(k-1)$. Such a state cannot lie in SEB($k-1$).  Therefore, the set of states corresponding to the matrices $\mathcal{M}$  consists a SUEB$k$. \qed

\vskip 5pt

Noticing that $\text{CH}_k\in W(k,k,k)$ for all integer $k \geq 2$. Therefore, by Theorem \ref{USEBk_weigh},  we arrive at the following corollary.

\begin{corollary}\label{USEBk_skew}
	Let $k$ be an integer such that $W(n,k,k+1)$ is nonempty for some integer $n$.    Then there exists some SUEBk with  numbers varying from $(d-k+1)d'$ to $dd'-1$  in $\mathbb{C}^{d}\otimes\mathbb{C}^{d'}$ whenever $d\geq 2k+1$ and    $d'\geq k+2$. 
\end{corollary} 

We should notice that the weighing matrices have been  studied by lots of researchers \cite{Berman78,Koukouvinos97,Arasu10,Best13,Kotsireas12,Leung11,Schmidt13}. For example,  there always exist some weighing matrix of the form $W(n,p^m,p^m+1)$ whenever $p^m>2$ for all prime $p$. In fact,    Gerald Berman proved a much more  strong result than  this \cite{Berman78}.
 
\begin{corollary}
	Let $p$ be a prime and $k=p^m>2$ for some positive integer $m$.   Then there exists some SUEBk with  numbers varying from $(d-k+1)d'$ to $dd'-1$  in $\mathbb{C}^{d}\otimes\mathbb{C}^{d'}$ whenever $d\geq 2k+1$ and    $d'\geq k+2$. 
\end{corollary}

\begin{corollary}
	Let $p_1,\cdots,p_s$ be different primes and $k=p_1^{m_1}\cdots p_s^{m_s}$ where    $m_1,\cdots,m_s$ are positive integers.  If $\gcd(p_i^{m_i}+1, k)=1$  for each $i=1,\cdots, s$,   Then there exists some SUEBk with  numbers varying from $(d-k+1)d'$ to $dd'-1$  in $\mathbb{C}^{d}\otimes\mathbb{C}^{d'}$ whenever $ d\geq k+\prod_{i=1}^s(p_i^{m_i}+1)$ and $d\geq 2k+1$ and   $d'\geq 2+ \prod_{i=1}^s(p_i^{m_i}+1)$.
\end{corollary}



\section{{Second Type  of suebk}}\label{fifth}
In the following, we try to construct a set of matrices $\mathcal{M}:=\{M_i\}_{i=1}^n$ which generates all the matrices of the left of the following form. While its complementary space   $\mathcal{M}^\perp$ is the set of matrices of the right of the following form where $r+s<k$.
\begin{figure}[h]
	\includegraphics[width=0.24\textwidth,height=0.22\textwidth]{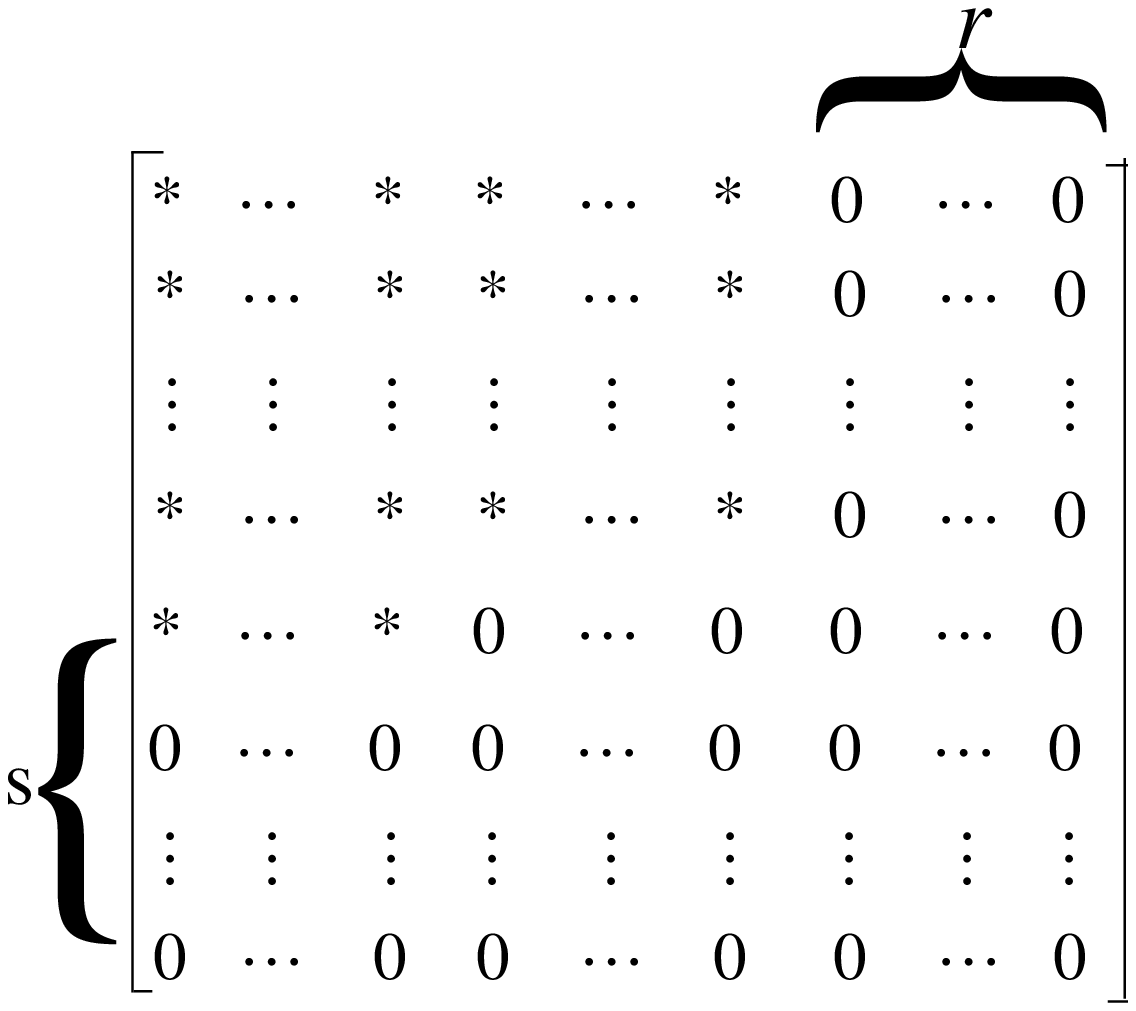}
	\includegraphics[width=0.235\textwidth,height=0.20\textwidth]{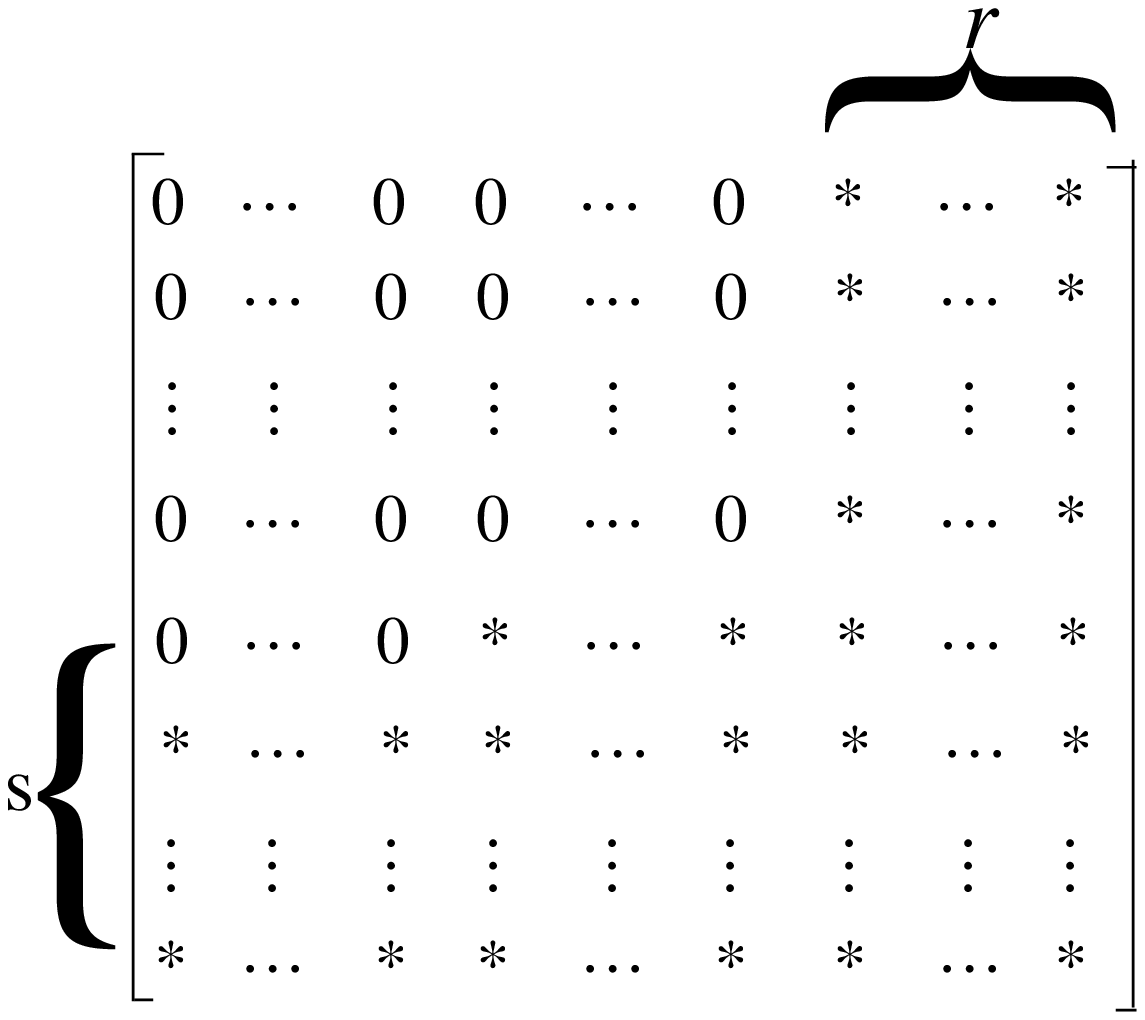}
	
\end{figure}

\begin{figure}[h]
	\includegraphics[width=0.23\textwidth,height=0.22\textwidth]{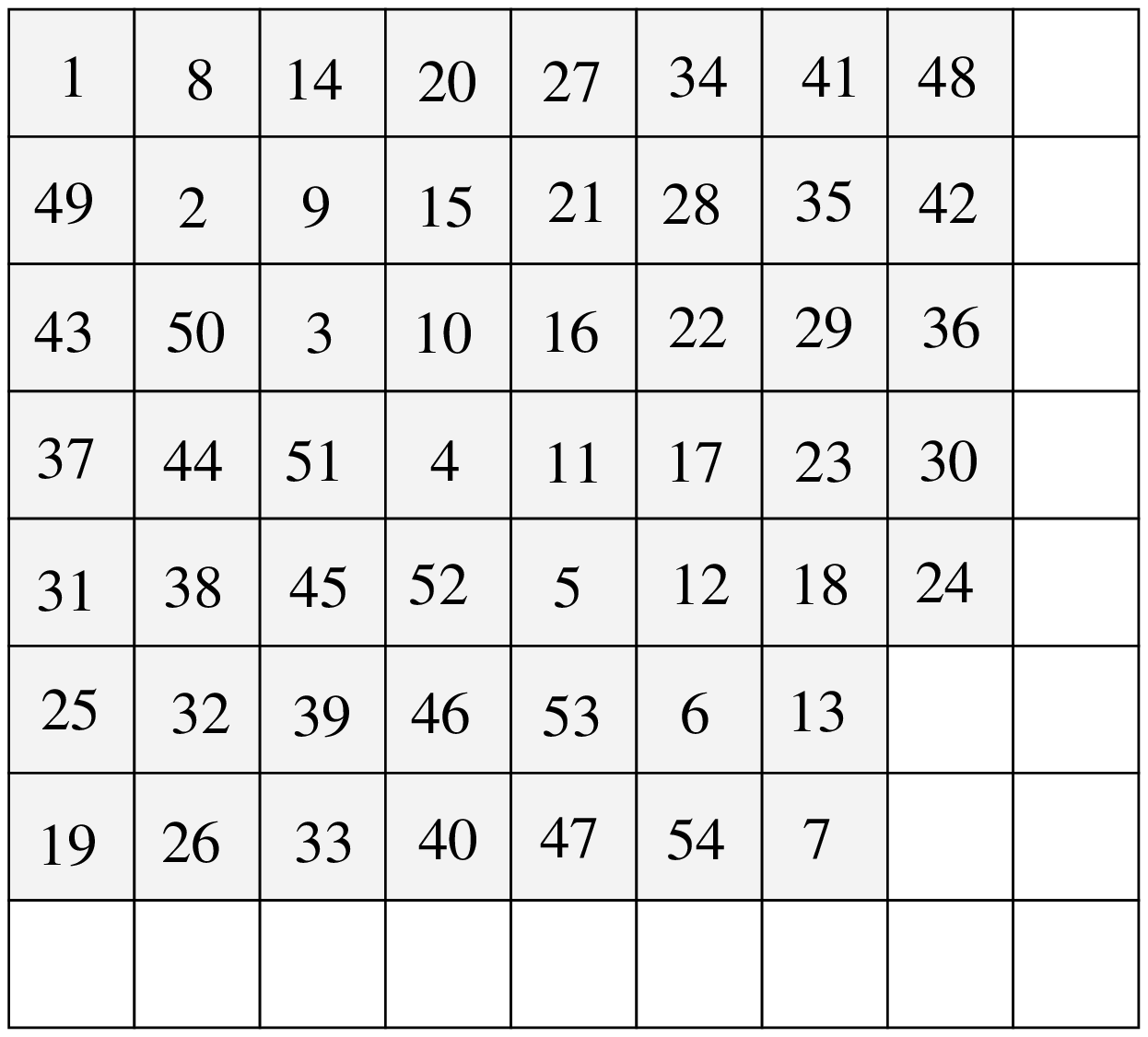}
	\includegraphics[width=0.23\textwidth,height=0.22\textwidth]{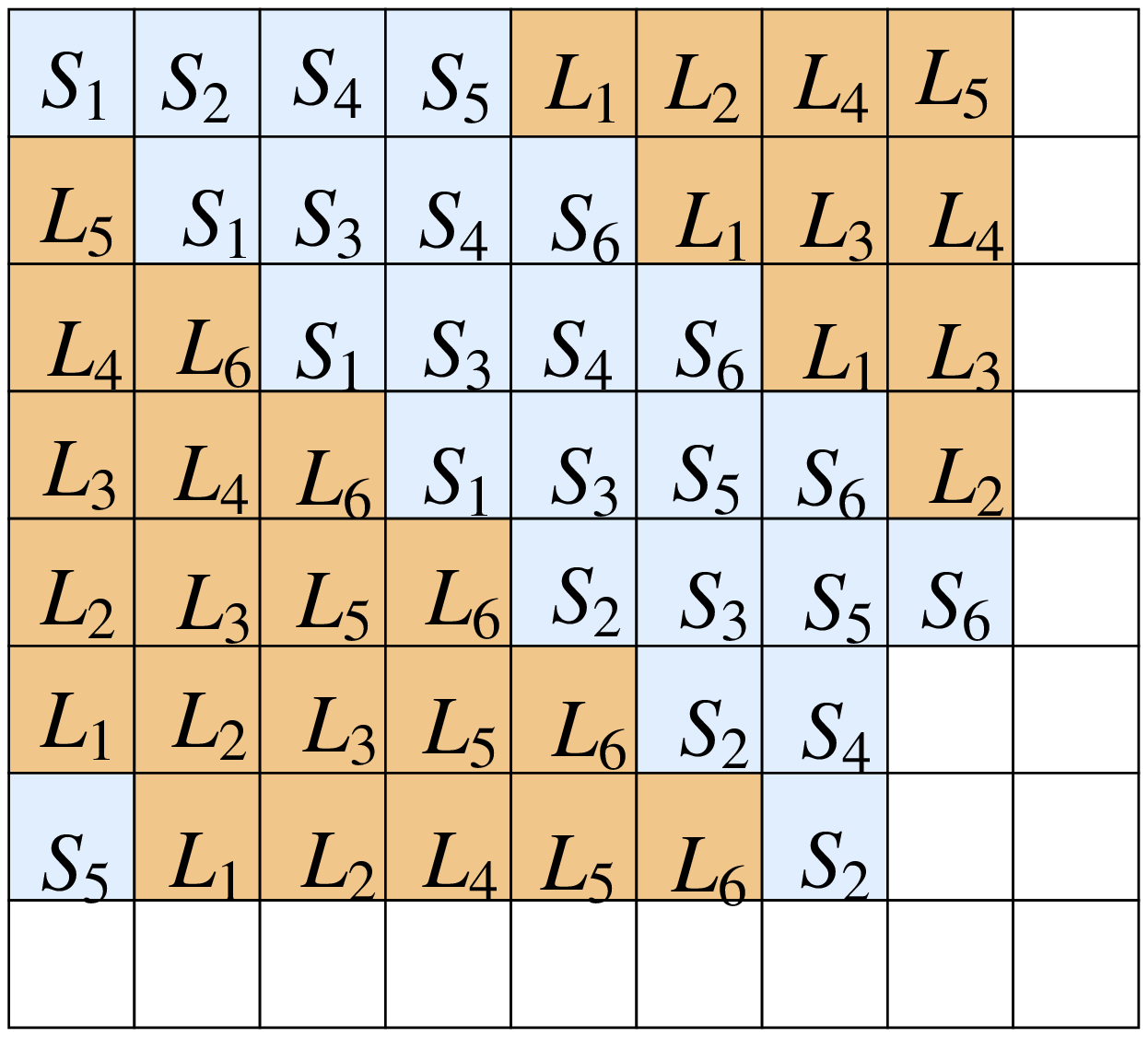}
	\caption{ The left figure shows the order of subset of $\mathcal{C}_{8\times 9}$. While the right hand one shows the distribution of the short and long states through this order. }\label{8times9Figure}
\end{figure}
\begin{example}\label{example2}
	 There exists  a SUEB4 in $\mathbb{C}^8\otimes \mathbb{C}^9 $ whose cardinality is $54$.
\end{example}

\noindent\emph{Proof.} As $54=7\times 8-2$,  we can define $\mathcal{B}_{54}$ to be the  set with 54 elements which can be obtained by    deleting    $\{ (6,8),(7,8)\}$   from $\mathcal{C}_{7\times 8}$.   Notice that any $6$ consecutive elements of $\mathcal{C}_{7\times 8}$ under the order   $\mathcal{O}_{7\times 8}$    come from different rows and columns.   Denote  $\mathcal{O}$ as the order of  $\mathcal{B}_{54}$ inherited from $\mathcal{O}_{7\times 8}$. As  $\mathcal{O}_{7\times 8}[(7,8)]=14,\mathcal{O}_{7\times 8}[(6,8)]=20$, any $5$ consecutive elements of $\mathcal{B}_{54}$ under the order   $\mathcal{O}$    come from different rows and columns (See the left figure of Fig. \ref{8times9Figure} for an intuitive view).   We have the following identity
\begin{equation}\label{Large_8times9_decom}
54=6\times 4+ 6\times 4.
\end{equation}
Since there are $54$ elements in the set $\mathcal{B}_{54}$, by the decomposition (\ref{Large_8times9_decom}), we can divide the set $\mathcal{B}_{54}$ into $(6+6)$ sets: $9$ sets (denote by $S_i,  1\leq i\leq 6$)  of cardinality $4$ and $6$ sets (denote by $L_j, 1\leq j\leq 6$)  of cardinality $6$. In fact, we can divide $\mathcal{B}_{54}$ into these $12$ sets through its order $\mathcal{O}$. That is,
$$\begin{array}{rl}
S_{i}:=&\{\mathcal{O}^{-1}[4(i-1)+x]\ \big| \ \ x=1,\cdots, 4\}, 1\leq i\leq 6,\\[2mm] 
L_{j}:=&\{\mathcal{O}^{-1}[24+5(j-1)+y]\ \big| \ \ y=1,\cdots,5\},  1\leq j\leq 6.
\end{array}$$
See the right hand side of Fig. \ref{8times9Figure} for an intuitive view of the set $S_i,L_j$.
Set  

$$O_5=\left[
\begin{array}{ccccc}
1&1&1&1&0\\[2mm]
1&w&w^2&0&1\\[2mm]
1&w^2&0&w&w^2\\[2mm]
1&0&w&w^2&w\\ [2mm]
0&1&w^2&w&w
\end{array}
\right], \ \ \text{ where } w=e^{2\pi \sqrt{-1}/3}.$$
We can easily check that
$
 O_5 O_5^\dagger= 4I_5.
$
Now set $v_x$ be the $x$-th row of $\text{CH}_4$  ($x=1,2,3,4$) and $w_y$ be the $y$-th row of $O_5$  ($y=1,2,3,4,5$). So $v_x\in \mathbb{C}^4$ and $v_y\in \mathbb{C}^5$. So we can construct the following $6\times 4+ 6\times 5=54$ matrices:
$$\begin{array}{c}
M_{8\times 9} (S_i, \frac{1}{\sqrt{3}}v_x), M_{8\times 9} (L_j,\frac{1}{\sqrt{3}}w_y), \\[2mm]
\ 1\leq i\leq 6,1\leq x\leq 4, 1\leq j\leq 6,1\leq y\leq 5.
\end{array}$$

Let $\mathcal{M}$ to be the set of the above matrices. Note that  the elements of each $S_i$  or $L_j$  are coordinately different.  Hence by  {\bf Observation} 1, the states correspond  to the above $N$ matrices  belong to  SES4.
Since $\text{CH}_4 \text{CH}_4^\dagger=4I_4$,  $v_1,v_2,v_3,v_4$ are pairwise orthogonal. Similarly, as $O_5 O_5^\dagger =4 I_5$, $w_1,w_2,w_3,w_4,w_5$ are also pairwise orthogonal.  And the $12$ sets above are pairwise disjoint.
Therefore, by Lemma \ref{Orthogonal}, the  $54$ matrices above are pairwise orthogonal.  Set $V$ be the linear subspace of $\text{Mat}_{8\times 9}(\mathbb{C})$. Each matrix in $\mathcal{B}_\perp:=\{E_{i,j}\in \text{Mat}_{8\times9}(\mathbb{C})| (i,j)\in \mathcal{C}_{8\times 9} \setminus \mathcal{B}_{54}\}$ is orthogonal to $V$.  And  the dimension of $\text{span}_\mathbb{C}(\mathcal{B}_\perp)$ is just $(72-54)$. Therefore, $V^\perp=\text{span}_\mathbb{C}(\mathcal{B}_\perp)$.  One should note that the rank of any matrix in $\text{span}_\mathbb{C}(\mathcal{B}_\perp)$ is less than $4$. That is to say, any state orthogonal to the states corresponding to $\mathcal{M}$ has Schmidt rank at most 3. Such a state cannot lie in SEB3.  Therefore, the set of states corresponding to the matrices $\mathcal{M}$  consists a SUEB4. \qed

\begin{theorem}\label{USEBk_weigh_2}
	Let $k$ be a positive integer. 	 Suppose there exist $a,b,m,n \in \mathbb{N}$ such that $W(m,k,a)$ and $W(n,k,b)$ are nonempty and $\gcd(a,b)=1$.  Let $d, d'$ be integers. If there are decompositions $d=m_1+s, d'=m_2+r$ such that $m_1,m_2\geq \max\{a,b\}+2$ and $1\leq r+s<k$. Then for any $N\in [m_1m_2, dd'-1]$, there exists  a SUEBk in $\mathbb{C}^d\otimes \mathbb{C}^{d'} $ whose cardinality is exactly $N$.
\end{theorem}

\noindent\emph{Proof.} Without loss of generality, we suppose $a< b$ and   $A\in W(m,k,a)$, $B\in W(n,k,b)$. Let $A_1,\cdots, A_a$ be the rows of $A$ and $B_1,\cdots, B_b$ be the rows of $B$. We separate the interval $[m_1m_2, d d' )$ into $s+t$ pairwise disjoint intervals: 
$$\begin{array}{cl}
&[(m_1+i)m_2,(m_1+i+1)m_2),\ \  0 \leq i\leq s-1,  \\ 
&[d(m_2+j),d(m_2+j+1)),\ \  0 \leq j\leq r-1. 
\end{array}
$$
 
 Any integer $N\in [m_1m_2, dd'-1]$ lies in one of the above $s+t$ intevals. Without loss of generality, we assume that $N\in [(m_1+i_0)m_2,(m_1+i_0+1)m_2)$ for some $i_0\in \{0,\cdots, s-1\} $.  Suppose $N=(m_1+i_0)m_2+f,$ with $0\leq f\leq m_2-1$.  Denote $\mathcal{B}_N$ to be the set by deleting the elements $\{ (m_1+i_0+1,i) | 1\leq i \leq m_2-f\}$  from  $\mathcal{C}_{(m_1+i_0+1)\times m_2}$. Then we have a coordinate set $\mathcal{C}_{(m_1+i_0+1)\times m_2}$  with order $\mathcal{O}_{(m_1+i_0+1)\times m_2}$. Notice that any $\max\{a,b\}$+1  consecutive elements  of $\mathcal{C}_{(m_1+i_0+1)\times m_2}$ under the order  $\mathcal{O}_{(m_1+i_0+1)\times m_2}$ are  coordinate different as $m_1,m_2\geq \max\{a,b\}+2$.   The subset $\mathcal{B}_N$ inherit an order $\mathcal{O}$ from that of $\mathcal{C}_{(m_1+i_0+1)\times m_2}$.     One can find that any  $a$ or $b$  consecutive elements  of $\mathcal{B}_N$  under the order   $\mathcal{O}$  come from different rows and columns. As $N\geq m_1m_2>(a-1)\times (b-1)$, by Lemma \ref{Frobenius}, there exist nonnegative integers  $s,t$  such that
\begin{equation}\label{Large_dtimesd_decom}
N=s\times a+ t\times b.
\end{equation}
Since there are $N$ elements in the set $\mathcal{B}_{N}$, by the decomposition (\ref{Large_dtimesd_decom}), we can divide the set $\mathcal{B}_{N}$ into $(s+t)$ sets: $s$ sets (denote by $S_i,  1\leq i\leq s$)  of cardinality $a$ and $t$ sets (denote by $L_j, 1\leq j\leq t$)  of cardinality $b$. In fact, we can divide $\mathcal{B}_{N}$ into these $s+t$ sets through its order $\mathcal{O}$. That is,
$$\begin{array}{rl}
S_{i}:=&\{\mathcal{O}^{-1}[(i-1){a}+x]\ \big| \   x=1,\cdots, a\}, 1\leq i\leq s,\\[2mm] 
L_{j}:=&\{\mathcal{O}^{-1}[{s}a+(j-1)b+y]\ \big| \  y=1,\cdots, b\},  1\leq j\leq t.
\end{array}$$

Then we can construct the following $s\times a+ t\times b=N$ matrices:
$$\begin{array}{c}
M_{d\times d'} (S_i, \frac{1}{\sqrt{k}}A_x), M_{d\times d'} (L_j,\frac{1}{\sqrt{k}}B_y) \\[2mm]
\ 1\leq i\leq s,1\leq x\leq a, 1\leq j\leq t,1\leq y\leq b.
\end{array}$$
Let $\mathcal{M}$ to be the set of the above matrices. Note that the $(s+t)$ sets $S_1,\cdots,S_s,L_1\cdots,L_t$ are pairwise disjoint. And the rows of $A$ (resp. $B$) are orthogonal to each other  as  $AA^\dagger= kI_a$ (resp. $BB^\dagger=kI_b$). By  Lemma \ref{Orthogonal}, the above $sa+tb$ matrices are orthogonal to each other. By construction, all the sets $S_1\dots,S_s,L_1,\cdots, L_t$ are  all coordinately different.  Using this fact and the  definition of generalized weighing matrices, the states corresponding to these matrices are all belong to SESk (see {\bf Observation} 1).  Set $V$ be the linear subspace of $\text{Mat}_{d\times d'}(\mathbb{C})$. Each matrix in $\mathcal{B}_\perp:=\{E_{i,j}\in\text{Mat}_{d\times d'}(\mathbb{C})| (i,j)\in \mathcal{C}_{d\times d'} \setminus \mathcal{B}_N\}$ is orthogonal to $V$.  And  the dimension of $\text{span}_\mathbb{C}(\mathcal{B}_\perp)$ is just $dd'-N$. Therefore, $V^\perp=\text{span}_\mathbb{C}(\mathcal{B}_\perp)$. As $r+s<k$, so the rank of any matrix in $\text{span}_\mathbb{C}(\mathcal{B}_\perp)$ is less than $k$. That is to say, any state orthogonal to the states corresponding to $\mathcal{M}$ has Schmidt rank at most $(k-1)$. Such a state cannot lie in SEB($k-1$).  Therefore, the set of states corresponding to the matrices $\mathcal{M}$  consists a SUEB$k$. \qed

\vskip 5pt

As application,  the   Theorem \ref{USEBk_weigh_2}  give us that there is some SUEB4 in  $\mathbb{C}^8\otimes \mathbb{C}^9 $ whose cardinality lies in $[49, 71]$,  where $a=4, b=5,$  $m_1=7,s=1, m_2=7,r=2.$

In fact, we may move further than the results showed in  Theorem \ref{USEBk_weigh} and  Theorem \ref{USEBk_weigh_2}. Here we present some examples (See Example \ref{example3}) which is beyond the scope of  Theorem \ref{USEBk_weigh} and  Theorem \ref{USEBk_weigh_2}.  But their proof  can be originated from the main idea of the two kinds  of  constructions of SUEBk.

\begin{example}\label{example3}
	For any integer $N\in [12,19]$, there exists  a SUEB3 in $\mathbb{C}^4 \otimes \mathbb{C}^5 $ whose cardinality is exactly $N$(See Fig. \ref{example4_5}). 
\end{example}

\begin{figure} 
	\includegraphics[width=0.5\textwidth,height=0.43\textwidth]{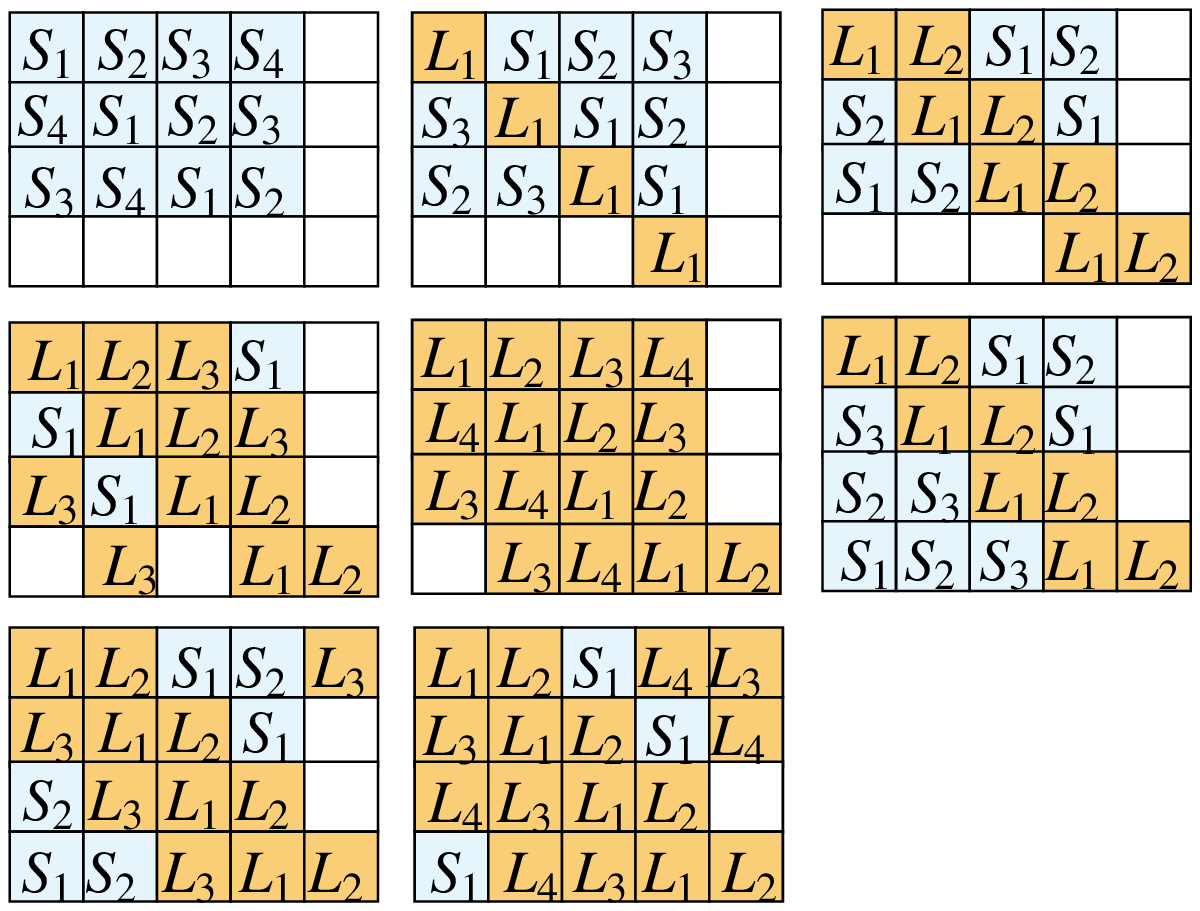}
	\caption{   This figure shows the distribution of the short states and long states for constructing SUEB3 in $\mathbb{C}^4\otimes\mathbb{C}^5$  with cardinality $N$ varying from $12$ to $19$.    }\label{example4_5}
\end{figure}

\section{conclusion and discussion}\label{sixth}
We present a method to construct the special unextendible entangle basis of type $k$.   The main idea here is to decompose the whole space into two subspaces such that the rank of one subspace easily bound by $k$ and the other can be generated by two kinds of the special entangled states of type $k$.  Here  the two kinds of  the special entangled states of type $k$ is related to a combinatoric concept which is known as  weighing matrices. This method is effective for $k=p^m\geq 3$. In these settings, we can obtain a series of SUEBk when the local dimensions are large.   In fact, based on two kinds of subspaces whose rank can be easily upper bounded by $k$, we give two types of constructions of the SUEBks.

However,  there are lots of unsolved cases.  Find out the largest linear subspace such that it do not contain any scpecial entangled states of type $k$. This is related to determine the minimal cardinality of possible SUEBk.  It is much more interesting to find some other methods that can solve the general existence of SUEBk.
\vskip 5pt

	\noindent{\bf Acknowledgments}\, \,  The author thank Mao-Sheng Li  for helpful discussion. This work is supported by the Research startup funds of DGUT
	with Grant No. GC300501-103.

\bigskip

\vskip 20pt

\end{document}